\documentclass[12pt]{article}

\usepackage{amssymb}
\usepackage{amsfonts}
\usepackage{amsmath}
\usepackage[english]{babel}
\usepackage{graphicx}
\usepackage{array}
\usepackage{vmargin}
\usepackage{listings}
\usepackage[mathcal]{euscript}
\usepackage{listings}
\usepackage{amsthm}
\usepackage{stmaryrd}
\usepackage{latexsym,amscd}

\setmarginsrb{2cm}{2cm}{2cm}{3cm}{0cm}{0cm}{0cm}{1cm}

{\end{array}}

\newtheorem{thm}{Theorem}
\newtheorem{prop}{Proposition}
\newtheorem{lem}{Lemma} 

\newtheorem{de}{Definition}
\newtheorem{cor}{Corollary}

\newtheorem{preremark}{Remark}
  \newenvironment{rk}{\begin{preremark}\rm}{\end{preremark}}

\newcommand{\R}{\mathbb{R}}

\newcommand{\N}{\mathbb{N}}

\renewcommand{\P}{\mathbb{P}}
\newcommand{\E}{\mathbb{E}}
\newcommand{\Om}{\Omega}
\newcommand{\om}{\omega}

\renewcommand{\d}[2]{\dfrac{\partial #1}{\partial #2}}
\newcommand{\dd}[3]{\dfrac{\partial^{#1} #2}{\partial #3^{#1}}}
\newcommand{\ddd}[3]{\dfrac{\partial^{2} #1}{\partial #2 \partial #3}}
\newcommand{\dt}{\frac{\mathrm{d}}{\mathrm{d}t}}

\newcommand{\ds}{\: \mathrm{d}s}

\renewcommand{\a}{\alpha}
\renewcommand{\b}{\beta}
\newcommand{\g}{\gamma}
\newcommand{\dl}{\delta}
\renewcommand{\k}{\kappa}
\newcommand{\Dl}{\Delta}
\newcommand{\eps}{\epsilon}
\renewcommand{\l}{\lambda}

\newcommand{\lb}{\llbracket}
\newcommand{\rb}{\rrbracket}

\newcommand{\htun}{h_t^{(1)}}
\newcommand{\htde}{h_t^{(2)}}
\newcommand{\ktun}{k_t^{(1)}}
\newcommand{\ktde}{k_t^{(2)}}
\newcommand{\hinfun}{h_{\infty}^{(1)}}
\newcommand{\hinfde}{h_{\infty}^{(2)}}
\newcommand{\httun}{\tilde h_t^{(1)}}
\newcommand{\httde}{\tilde h_t^{(2)}}

\newcommand{\Httun}{\tilde H_t^{(1)}}
\newcommand{\Httde}{\tilde H_t^{(2)}}

\newcommand{\G}{\Gamma}
\newcommand{\Ginf}{\G_{\infty}}
\newcommand{\bsl}{\backslash}

\newcommand{\Xn}{X^{(n)}}

\newcommand{\Cn}{C^{(n)}}
\newcommand{\An}{A^{(n)}}
\newcommand{\Bn}{B^{(n)}}

\newcommand{\Mnp}{M^{p_0,(n)}}
\newcommand{\Mn}{M^{(n)}}
\newcommand{\Gn}{G^{(n)}}

\newcommand{\MAn}{M^{A,(n)}}
\newcommand{\MBn}{M^{B,(n)}}

\newcommand{\D}{\mathbb{D}}
\newcommand{\DE}{\D([0,+ \infty),E)}
\renewcommand{\DH}{\D([0,+ \infty),H)}

\renewcommand{\le}{\lambda_{\eta}}
\newcommand{\len}{\le^{(n)}}
\newcommand{\Dp}{\Delta_{p,p'}}

\newcommand{\cinf}{c_{\infty}}
\newcommand{\ginf}{g_{\infty}}

\makeatletter
\def\maketitle{%
  \null
  \thispagestyle{empty}%
  \begin{center}\leavevmode
    \normalfont
    {\LARGE \@title\par}%
    \vskip 0.6cm
    {\large \@author\footnote{E-mail address: \@email.}\par}%
    \vskip 0.6cm
    {\footnotesize \it \@location\par}%
  \end{center}%
  \null
  }
\def\location#1{\def\@location{#1}}
\def\email#1{\def\@email{#1}}
\makeatother

\begin{document}

\title{A model for coagulation with mating}
\author{Raoul Normand}
\location{Laboratoire de Probabilit\'es et Mod\`eles Al\'eatoires, UPMC, \\ 175 rue du Chevaleret, 75013 Paris, France}
\email{raoul.normand@upmc.fr}
\maketitle

\begin{abstract}
We consider in this work a model for aggregation, where the coalescing particles initially have a certain number of potential links (called arms) which are used to perform coagulations. There are two types of arms, male and female, and two particles may coagulate only if one has an available male arm, and the other has an available female arm. After a coagulation, the used arms are no longer available. We are interested in the concentrations of the different types of particles, which are governed by a modification of Smoluchowski's coagulation equation --- that is, an infinite system of nonlinear differential equations. Using generating functions and solving a nonlinear PDE, we show that, up to some critical time, there is a unique solution to this equation. The Lagrange Inversion Formula allows in some cases to obtain explicit solutions, and to relate our model to two recent models for limited aggregation. We also show that, whenever the critical time is infinite, the concentrations converge to a state where all arms have disappeared, and the distribution of the masses is related to the law of the size of some two-type Galton-Watson tree. Finally, we consider a microscopic model for coagulation: we construct a sequence of Marcus-Lushnikov processes, and show that it converges, before the critical time, to the solution of our modified Smoluchowski's equation.
\end{abstract}

\begin{scriptsize}

\end{scriptsize}

\section{Introduction}

In this work, we study a model for coagulation of particles, generalizing the original model of Smoluchowski \cite{Smolu}, and a recent model of Bertoin \cite{B}. We consider particles which are initially given a certain number of male and female arms. These arms are used to perform the coagulations: two particles coagulate when a male arm of one and a female arm of another bind. This can be used to model the formation of polymers. For instance, consider male particles (which have only male arms), and female particles. Then a coagulation between a male and a female particle can be thought of as an ionic bond between a cation and a anion. This kind of models has also been investigated in the physical literature. For instance, in \cite{Redner1}, \cite{Redner2}, the authors study coalescing monomers with two types, A and B, with bonding only allowed between A and B, hence forming alternating linear polymers. In this work, this corresponds to giving to each particle exactly one male arm and one female arm.

In our model, a particle is characterised by a triple $(a,b,m)$, $a \in \N$ being its number of male arms, $b \in \N$ its number of female arms, and $m \in \N^*$ its mass. Two particles may coagulate when one has an available female arm and the other has an available male arm, and when a coagulation occurs, the used arms disappear. Hence, we may only observe the transition
\[
\{(a,b,m),(a',b',m')\} \to (a+a'-1,b+b'-1,m+m').
\]
We will assume that this transition occurs with a rate given by the number of pairs formed of a female arm and of a male arm, that is $a'b+ab'$. We wish to study how the concentration of each type of particle evolves when time passes. The precise mathematical formulation is given in Section 2.

This model is a modification of the well-known model of Smoluchowski \cite{Smolu}. Recall that Smoluchowski's coagulation equations \cite{Smolu} describe the evolution of the concentrations of particles in a medium, where particles are characterised only by their masses. When two particles of masses $m$ and $m'$ coagulate, they merge into a single particle of mass $m+m'$. Such a coagulation occurs with rate $\k(m,m')$, where $\k$ is some symmetric nonnegative kernel. In Smoluchowski's original model, the masses are assumed to be positive integers. The concentration $c_t(m)$ of particles of mass $m$ is governed by the following infinite system of nonlinear differential equations
\[
\dt c_t(m)=\frac{1}{2} \sum_{m'=1}^{m-1} c_t(m')c_t(m-m') \k(m',m-m') - c_t(m) \sum_{m'=1}^{+\infty} c_t(m') \k(m,m'),
\]
for $m \in \N^*$. The first term accounts for the creation of particles of mass $m$ by coagulation of particles of mass $m'$ and $m-m'$; the second for disappearance of particles of mass $m$ by coagulation with other particles.

For general kernels $\k$, explicit solutions are not known. However, some have been obtained in different cases, notably whenever the kernel is constant \cite{Smolu}, additive \cite{Gol} or multiplicative \cite{Leod}. In the multiplicative case, solutions are obtained up to a critical time, known as the gelation time. This is interpreted as the time when a particle of infinite mass appears. It absorbs some of the particles, and the total mass starts to decrease.

Smoluchowski's equation (and some variations) have been extensively studied, both from an analytical (e.g. \cite{CdC}, \cite{EMP}, \cite{Rezaana}) and a probabilistic point of view (e.g. \cite{DT}, \cite{Rezaproba}, \cite{Jeon}, \cite{Norris}, and see also the review by Aldous \cite{Aldous}). In general, little is known after the gelation time, and most results are obtained before (see however \cite{FG}, \cite{FL}). The existence and uniqueness of a solution before gelation has been obtained only in 1999 by Norris \cite{Norris}, under the assumption that $\k$ is sublinear, i. e. $\k(n,m)/(nm)$ is bounded.

From a probabilistic point of view, some microscopic models have been studied, beginning with Marcus \cite{Marcus} and Lushnikov \cite{Lushnikov}. Heuristically, one considers a finite number of particles, and each couple of particles with masses $m$ and $m'$ coalesces with rate $\k(m,m')$. After suitable change of time and renormalization, one expects this system to converge to a solution of Smoluchowski's equation. This has been shown by Jeon in 1998 \cite{Jeon} (up to extraction of a subsequence), provided the rate is strictly sublinear (i.e. $\k(n,m)/n \to 0$ when $n \to + \infty$). In particular, there is no gelation in this case. Norris \cite{Norris} extended his results one year later by showing the convergence of the model before gelation, whenever the rate is sublinear. Other points of view are also considered; e.g. in \cite{Rezaproba}, the authors show that coagulating Brownian particles form clusters whose size evolves according to Smoluchowski's equation.

An interesting question is to deal with the case when the coagulations are restricted by some device. Typically, one may think of covalent bonds: a given atom can only perform a given number of bonds. In this direction, Bertoin \cite{B} studied two models where a particle is characterised by its number of arms and by its mass, and it uses its arms to perform aggregations. The concentrations of each type of particle is governed by a modification of Smoluchowski's equation. In \cite{B}, he obtains solution up to some time $T$, and shows that whenever gelation does not occur (i.e. $T=+ \infty$), there is a limit state where all the arms have disappeared: the concentrations converge to limiting concentrations which bear a striking resemblance with the law of the size of some Galton-Watson tree. This fact is explained in \cite{B2} and \cite{B3}. It is also worth noticing that Bertoin's model can be related to Smoluchowski's for the constant, additive and multiplicative kernels. We will also see that our sexed model contains Bertoin's: the oriented model corresponds indeed to ours if each particle is given precisely one female arm, and the symmetric model corresponds to the sexed one if the particles are given a gender uniformly at random.

This paper is divided in two parts. In the first one (Sections 2 to 5), we shall study the sexed Smoluchowski's equation, which is an infinite system of nonlinear differential equations. We first (Section 2) introduce the problem, and prove some physically intuitive facts. Then (Section 3), we prove our main result: up to some critical time, there exists a unique solution to the system, and its moment generating function can be expressed explicity in terms of the initial data. The tools used are analoguous to those in \cite{B}, but since we are dealing with a two-dimensional problem, several technical issues need to be addressed. The outline of the proof is as follows. First, we transform the system into a PDE problem by considering the generating functions of the concentrations. This PDE is not quasilinear, but it may however be solved by the method of characteristics. This method requires the inversion of a two-dimensional mapping, and this can be done precisely up to the critical time. Unfortunately, even for monodisperse initial conditions (i.e. there are only particles of mass 1 at time 0), the inversion is not explicit (one could use the two-variable Lagrange inversion formula, but in general, the expression it provides is too cumbersome). Nonetheless, in some specific cases (Section 4), the Lagrange Inversion Formula yields explicit results. In particular, we recover the solutions obtained in \cite{B}. Finally, we show (Section 5) that there exist limiting concentrations when $t \to + \infty$, and that they are related to the distribution of the total progeny of some two-type Galton-Watson process.

In the second part (Section 6), we study a microscopic model. Given a finite number of particles, we let them coagulate and observe the evolution of the concentrations of the different types of particles. This is a Marcus-Lushnikov process, and we show that it converges, before the critical time, to a process solving Smoluchowski's equation \eqref{syst}. As pointed out earlier, this kind of convergence had already been proved by Norris (\cite{Norris}, see as well \cite{Jeon}). The difference here is that we consider a model with male or female arms. Moreover, the proof is made much easier by the fact that the rate of coagulation is explicit. In particular, we will appeal to the PDE obtained in the first part. This discrete model provides a justification to the sexed Smoluchowski's equation \eqref{syst}.

Finally, note that our construction can also provide a model for random oriented graphs, called the configuration model, since a coagulation can be seen as the creation of an oriented edge between two vertices in a graph, whose orientation is given e.g. from the male arm to the female arm. Hence, we can consider a large number $n$ of particles and let them coagulate. When all the coagulations are performed, we obtain a set of oriented graphs. When $n \to + \infty$, we may wonder what the distribution of their sizes is, what a typical graph looks like, etc. A heuristic answer, motivated by the works \cite{B2}, \cite{B3}, and by the results obtained in this paper (Section 5), is that a typical graph would be a two-type Galton-Watson tree (with the convention of orientation above), provided there are few arms (with the notations of this paper, this means $T_c=+ \infty$ and $\mu$ is not degenerate).

\section{Setting and results}

\subsection{Notations}

Let us first introduce some notations and Smoluchowski's equation, and state our main result.
\begin{itemize}
\item $\N=\{0,1,2,\dots\}$ and $\N^*=\{1,2,\dots\}$.
\item $S=\N \times \N \times \N^*$ is the set of the different types of particles. A generic element of $S$ will be denoted by $p$, and if $p=(a,b,m)$, we will call a $p$-particle a particle with $a$ male arms, $b$ female arms, and mass $m$.
\item For $p=(a,b,m) \in S$ and $p'=(a',b',m') \in S$, we will denote
\[
p.p'=a'b+ab'
\]
the \textit{rate} of coagulation and
\[
p \circ p'=(a+a'-1,b+b'-1,m+m')
\]
the type of the particle resulting from such a coagulation. We say that $p' \preceq p$ if $a' \leq a+1$, $b' \leq b+1$ and $m' \leq m-1$. When $p' \preceq p$, we write
\[
p \bsl p' = (a+1-a',b+1-b',m-m')
\]
the type of particle such that $p' \circ (p \bsl p')=p$.
\item For two functions $c, f : S \to \R$, we will denote, when the series converge absolutely,
\[
\langle c,f \rangle := \sum_{p \in S} c(p) f(p).
\]
When using this notation, we will write, with a slight abuse of notation, $a$ for the function $(a,b,m) \mapsto a$, $b$ for $(a,b,m) \to b$, etc.
\end{itemize}
Let us recall our goal. We are interested in a system of coagulating particles with male and female arms. We assume that each couple formed of a $p$-particle and of a $p'$-particle coagulate with rate $p.p'$, to form a $p \circ p'$-particle. This means that if we denote $c_t(p)$ the concentration of $p$-particles then $(c_t(p),p \in S)$ solves the following infinite system of nonlinear differential equations
\begin{equation}
\dt c_t(p)= \frac{1}{2}\sum_{p' \preceq p} p'.(p \bsl p') c_t(p')c_t(p \bsl p') - c_t(p) \sum_{p' \in S} p.p' c_t(p').
\label{syst}
\end{equation}
The first term accounts for the creation of $p$-particles by coagulation of $p'$- and $p \bsl p'$-particles (the factor $1/2$ comes from an obvious symmetry). The second accounts for the disappearence of $p$-particles by coagulation with other particles. Let us once write down this formula explicitly. For all $(a,b,m) \in S$, the concentration of $(a,b,m)$-particles verifies
\[
\begin{split}
\dt c_t(a,b,m) = & \frac{1}{2} \sum_{m'=1}^{m-1} \sum_{a'=0}^{a+1} \sum_{b'=0}^{b+1} (a'(b+1-b')+b'(a+1-a')) \times \\
 &  c_t(a',b',m')c_t(a+1-a',b+1-b',m-m') \\
 & - \sum_{m' \geq 1} \sum_{a' \geq 0} \sum_{b' \geq 0} (ab'+a'b) c_t(a,b,m) c_t(a',b',m').
\end{split}
\]
Let us now define what we call a \textit{solution} to Smoluchowski's equation.
\begin{de}
We call a family $\left ( (c_t(p))_{p \in S}, t \in [0,T) \right )$ of differentiable functions a solution of Smoluchowski's equation (or of system \eqref{syst}), if
\begin{enumerate}
\item For every $t \in [0,T)$, $\langle a + b, |c_t| \rangle < + \infty$,
\item $\langle a^2 + b^2, |c_t| \rangle < + \infty$ for $t$ in a neighbourhood of $0$,
\item The family $(c_t(p))$ solves the system \eqref{syst} for $t \in [0,T)$.
\end{enumerate}
\end{de}

\begin{rk}
\begin{itemize}
\item We will always assume that at time 0, $\langle a + b +1 , c_0 \rangle < + \infty$, and that the mean number of male arms $\langle a,c_0 \rangle$ and the mean number of female arms $\langle b, c_0 \rangle$ are equal. Physically, it is then obvious that they will remain equal as time passes. This shall be proven later on, in Lemma \ref{atbt}.
\item It is easy to see that if $(c_t)_{t \in [0,T)}$ is a solution to \eqref{syst} with initial conditions $c_0$, and $\lambda > 0$, then $(\lambda c_{t/\lambda^2})_{t \in [0,T)}$ is a solution to \eqref{syst} with initial conditions $\lambda c_0$. Hence, it is enough to assume that $\langle a, c_0 \rangle = \langle b, c_0 \rangle = 1$, what will always be the case from now on.
\end{itemize}
\end{rk}

\subsection{Main result}

Our main result is existence and uniqueness of a solution to \eqref{syst} up to a critical time. In all the statements and proofs, we are given nonnegative initial concentrations $c_0$ such that $\langle 1 ,c_0 \rangle < + \infty$, $\langle a,c_0 \rangle = \langle b, c_0 \rangle = 1$ and $\langle a^2+b^2, c_0 \rangle < + \infty$. We can then define the critical time $T_c$.
\begin{de}
Let
\[
M=\langle ab,c_0 \rangle + \sqrt{ \langle a^2 - a ,c_0 \rangle \langle b^2 - b ,c_0 \rangle }
\]
and
\[
T_c =
\left \{
\begin{array}{ll}
+\infty & \qquad \mathrm{if} \;\; M \leq 1 \\
\frac{1}{M-1} & \qquad \mathrm{if} \;\; M > 1.
\end{array}
\right.
\label{Tc}
\]
\end{de}

\noindent We will also constantly use the generating function of $(c_0)$
\[
g_0(x,y,z) := \sum_{(a,b,m) \in S} c_0(a,b,m) x^a y^b z^m.
\]
Since $\langle 1, c_0 \rangle < + \infty$, $g_0$ is well-defined on $[0,1]^3$. Using the assumption $\langle a + b , c_0 \rangle =1$, and e.g. monotone convergence, we also see that its partial derivatives with respect to $x$ and $y$ are well-defined and continuous on $[0,1]^3$ . For the same reason, they remain in $[0,1]$. We shall prove the following result.

\begin{thm}
\begin{description}
\item[(i)] Smoluchowski's equation \eqref{syst} with initial conditions $c_0$ has a unique solution $(c_t)$ defined on $[0,T_c)$.
\item[(ii)] For $t \in [0,T_c)$, $\langle a^2+b^2, c_t \rangle < + \infty$, and $\langle a^2+b^2, c_t \rangle \to + \infty$ when $t \to T_c$.
\item[(iii)] For $t \in [0,T_c)$ and $z \in [0,1]$, the mapping $\phi_t(.,.,z)$, given for $(x,y) \in [0,1]^2$ by 
\[
\phi_t(x,y,z) = \left ( (1+t)x-t \d{g_0}{y}(x,y,z), (1+t)y-t \d{g_0}{x}(x,y,z) \right ),
\]
has a right inverse $h_t = (\htun,\htde)$ which is well-defined and analytic on $(0,1)^2$. Then the generating function $g_t$ of $(c_t)$ is given by
\begin{equation}
g_t(x,y,z) = \frac{1}{1+t} \left ( \Httde(x,y,z)+ \Httun(0,y,z) \right ) + G_t(z).
\label{gtxyz}
\end{equation}
where for $t > 0$
\begin{equation}
\httun = \frac{1+t}{t}h_t^{(1)}(x,y,z) - \frac{x}{t} \qquad ; \qquad \httde := \frac{1+t}{t}h_t^{(2)}(x,y,z) - \frac{y}{t} 
\label{defhtt}
\end{equation}
and
\begin{itemize}
\item $\Httun$ is the antiderivative of $\httun$ with respect to $y$, vanishing at $y = 0$,
\item $\Httde$ is the antiderivative of $\httde$ with respect to $x$, vanishing at $x = 0$,
\item $G_t(z)$ is the antiderivative of 
\begin{equation}
\d{g_0}{z} \left ( h_t^{(1)}(0,0,z),h_t^{(2)}(0,0,z),z \right )
\label{Gt}
\end{equation}
with respect to $z$, vanishing at 0.
\end{itemize}
\item[(iv)] The total mass $\langle m , c_t \rangle$ is constant on $[0,T_c)$.
\end{description}
\label{grosthm}
\end{thm}

\subsection{Preliminary results}

In this section, we give some physically intuitive results, and deduce the ``weak'' form of the equation. Let us start with the following lemma (recall that $c_t(p)$ is meant to model a concentration).

\begin{lem}
Any solution to Smoluchowski's equation remains nonnegative, i.e. if $(c_t)_{t \in [0,T)}$ is a solution to \eqref{syst}, then for all $t \in [0,T)$ and $p \in S$, $c_t(p) \geq 0$.
\label{possol}
\end{lem}

\begin{proof}
Take some $t \in [0,T)$. System \eqref{syst} gives
\[
\dt c_t(a,b,1) = - c_t(a,b,1) \sum_{m' \geq 1} \sum_{a' \geq 0} \sum_{b' \geq 0}  (ab'+a'b) c_t(a',b',m') :=  - \g(t) c_t(a,b,1).
\]
Let $G(t)=\int_0^t \g(s) \; \mathrm{d}s$. Then $c_t(a,b,1)=c_0(a,b,1) e^{-G(t)}$, so it remains nonnegative. Let now $m \geq 1$, and suppose that the $c_t(a,b,m')$ are nonnegative for $a,b \geq 0$ and $1 \leq m' \leq m$. For some $p=(a,b,m+1)$, we have
\begin{align*}
\dt c_t(p)= & \frac{1}{2}\sum_{p' \preceq p} p'.(p \bsl p') c_t(p')c_t(p \bsl p') - c_t(p) \sum_{p' \in S} p.p' c_t(p') \\
 = & \b(t) - c_t(p) \g(t).
\end{align*}
So we may write
\[
c_t(p)=\left ( c_0(p) + \int_0^t \b(s) e^{G(s)} \ds \right )e^{-G(t)}.
\]
But $\b(t) \geq 0$ since it is a linear combination with nonnegative coefficients of the $c_t(a,b,m')$ for $a,b \geq 0$ and $1 \leq m' \leq m$. So $c_t(a,b,m+1)$ is nonnegative, what gives the result by induction.
\end{proof}

The following two lemmas are straightforward generalizations of Lemma 1 and 2 in \cite{B}. Note however that the monotone convergence used in the proofs requires that the coefficients $(c_t)$ be nonnegative.

\begin{lem}
\begin{description}
\item[(i)] If $(c_t)$ is a solution to Smulochowski's equation \eqref{syst}, then $t \mapsto \langle 1,c_t \rangle$, $t \mapsto \langle a,c_t \rangle$ and $t \mapsto \langle b,c_t \rangle$ are decreasing.
\item[(ii)] A family $(c_t)$ is a solution to \eqref{syst} if and only if it solves
\begin{equation}
\dt \langle c_t,f \rangle = \frac{1}{2} \sum_{p,p' \in S} p.p' c_t(p) c_t(p') (f(p \circ p') - f(p) -f(p'))
\label{eqcroch}
\end{equation}
for every bounded $f : S \to \R$.
\end{description}
\label{formfaible}
\end{lem}

\begin{rk}
\begin{itemize}
\item The derivative in this lemma has to be understood in the weak sense, i.e. the formula actually holds in the integral form. But if $f(a,b,m) \to 0$ when $(a,b,m) \to \infty$, then it is easy to check that the formula holds in the strong sense.
\item Consider in particular, the generating function of $c_t$, $g_t(x,y,z)=\langle x^a y^b z^m , c_t \rangle$. Then $g(.,.,z)$ is regular, in the sense of Definition \ref{defreg} below.
\end{itemize}
\label{rkgt}
\end{rk}

\begin{de}
We say that a function $(t,x,y) \mapsto g_t(x,y)$ defined on $[0,T) \times (0,1)^2$ is regular if
\begin{itemize}
\item $t \mapsto g_t(x,y)$ is $C^1$ and $(x,y) \mapsto \d{g_t}{t}(x,y)$ are $C^1$,
\item $(x,y) \mapsto g_t(x,y)$ is $C^2$, $t \mapsto \d{g_t}{x}(x,y)$ and $t \mapsto \d{g_t}{y}(x,y)$ are $C^1$.
\end{itemize}
\label{defreg}
\end{de}

\begin{lem}
Let $(c_t)$ be a solution to Smoluchowski's equation, and let
\[
\G_r=\inf \{t \geq 0, \langle a^2 + b^2,c_t \rangle > r \} \qquad \mathrm{and} \qquad \Ginf=\sup_{r > 0} \G_r.
\]
Consider the mean numbers of male and female arms $A_t = \langle a ,c_t\rangle$ and $B_t = \langle b ,c_t \rangle$, and assume $A_0 = B_0 =1$. Then
\begin{equation}
A_t = B_t = \frac{1}{1+t}
\label{eqatbt}
\end{equation}
for all $t \in [0,T \wedge \Ginf)$.
\label{atbt}
\end{lem}

\section{Proof of the theorem}

\subsection{Overview of the method}

In this section, we give a sketch of the proof which contains all the important ideas. The rigorous proof however requires some care, and it is given in detail afterwards. So, consider a solution $(c_t)_{t \in [0,T)}$ to Smoluchowski's equation \eqref{syst}, and
\[
g_t(x,y,z) = \langle x^a y^b z^m , c_t \rangle = \sum_{a \geq 0} \sum_{b \geq 0} \sum_{m \geq 1} c_t(a,b,m) x^a y^b z^m.
\]
Using \eqref{eqcroch} and Lemma \ref{atbt}, it is easy to see that $g_t$ solves the following PDE
\begin{equation}
\d{g_t}{t}=\d{g_t}{x} \d{g_t}{y}-\frac{1}{1+t} \left ( x \d{g_t}{x} + y \d{g_t}{y} \right ).
\end{equation}
Now, we can solve this PDE using the method of characteristics: we want to find a trajectory $(x(t),y(t))$ starting from some $(x,y) \in [0,1]^2$ such that $g_t(x(t),y(t),z)$ is easy to compute. So let
\[
(p_1(t),p_2(t))= \left ( \d{g_t}{x}(x(t),y(t),z),\d{g_t}{y}(x(t),y(t),z) \right ).
\]
An easy calculation shows that
\begin{equation}
\dot{p_1}(t) = \dd{2}{g_t}{x} \left ( \dot{x}(t) + p_2(t) - \frac{x(t)}{1+t} \right ) + \frac{\partial^2 g_t}{\partial x \partial y} \left ( \dot{y}(t)+p_1(t)-\frac{y(t)}{1+t} \right ) - \frac{p_1(t)}{1+t},
\end{equation}
and a similar formula for $\dot{p_2}$. Now, if we require
\[
\dot{x}(t) + p_2(t) - \frac{x(t)}{1+t}  = \dot{y}(t)+p_1(t)-\frac{y(t)}{1+t} = 0,
\]
then
\[
\dot{p_i}(t)=-\frac{p_i(t)}{1+t}, \qquad i=1,2.
\]
These ODE's are readily solved, with $p_1(0)=\d{g_0}{x}(x,y)$ and $p_2(0)=\d{g_0}{y}(x,y)$, and we obtain
\begin{equation}
p_i(t) = \frac{p_i(0)}{1+t}
\label{pi}
\end{equation}
and
\[
x(t)=x+(x-p_2(0))t \quad ; \quad y(t)=y+(y-p_1(0))t.
\]
Using the PDE, we now see that
\begin{equation}
\dt g_t(x(t),y(t),z) = - \frac{p_1(0)p_2(0)}{(1+t)^2},
\label{ztpoint}
\end{equation}
so by integrating 
\begin{equation}
g_t(x(t),y(t),z)=g_t(\phi_t(x,y,z),z)=g_0(x,y,z)-\frac{t}{1+t} \d{g_0}{x}(x,y,z) \d{g_0}{y}(x,y,z).
\label{zt}
\end{equation}
To obtain $g_t$, it only remains to invert $\phi_t$, for, if $\phi_t(h_t)=\mathrm{Id}$, then
\[
g_t(x,y,z) = g_0(h_t(x,y,z))-\frac{t}{1+t} \d{g_0}{x}(h_t(x,y,z)) \d{g_0}{y}(h_t(x,y,z)).
\]

We may now start a rigorous proof, which consists mainly of 3 steps: study the map $\phi_t$, then solve the PDE \eqref{PDE}, and show that the generating function of a family $(c_t)$ solves \eqref{PDE} if and only if $(c_t)$ solves Smoluchowski's equation \eqref{syst}. The conclusion is then easy to obtain.

\subsection{Inversion of the mapping}

In this section, we study the map $\phi_t$, which is useful both for solving theorically the PDE, and for obtaining explicit solutions. We will need two preliminary lemmas.

\begin{lem}
Let $\a >0$, $\b,\g \geq 0$ and $K=[0,\a] \times [0,\b] \times [0,\g]$. For $(r,s,t) \in K$, denote
\[
A(r,s,t):=
\begin{pmatrix}
r & s \\
t & r
\end{pmatrix}.
\]
Then for every $\eps >0$, there is a norm $\|.\|$ on $\R^2$ such that
\[
\max_{(r,s,t) \in K} \| A(r,s,t)\| \leq \a + \sqrt{\b \g} + \eps,
\]
where we also denote by $\|.\|$ the induced norm on the $2 \times 2$ matrices.
\label{lemnorme}
\end{lem}

\begin{rk}
This is a uniform version of the well-known result (see e.g. \cite{Serre}) which states that
\begin{itemize}
\item For every (square) matrix $A$ and norm $\|.\|$, one has $\|A\| \geq \rho(A)$, where $\rho(A)$ is the spectral radius of $A$,
\item For every matrix $A$ and $\eps > 0$, there is a norm $\|.\|$ such that $\|A\| \leq \rho(A)+\eps$.
\end{itemize}
Note indeed that $\a + \sqrt{\b\g}$ is the spectral radius of $A(\a,\b,\g)$.
\end{rk}

\begin{proof}
\begin{enumerate}
\item First assume that $\b$ and $\g$ are positive. We can diagonalize $A:=A(\a,\b,\g)$. If we let $a:=\a$, $b:=\sqrt{\b}$ and $c:=\sqrt{\g}$ then
\[
A=P
\begin{pmatrix}
a + bc & 0 \\
0 & a - bc
\end{pmatrix}
P^{-1},
\]
where
\[
P =
\begin{pmatrix}
b & -b \\
c & c
\end{pmatrix}
\qquad ; \qquad
P^{-1} = \frac{1}{2bc}
\begin{pmatrix}
c & b \\
-c & b
\end{pmatrix}.
\]
Now, consider the following norm: for $x \in \R^2$, let $\|x\|=\|P^{-1}x\|_{\infty}$, where $\|(x_1,x_2)\|_{\infty}=\max(|x_1|,|x_2|)$. Then for any $2 \times 2$ matrix $M$,
\[
\|M\|=\max_{x \neq 0} \frac{\|Mx\|}{\|x\|} = \max_{x \neq 0} \frac{\|P^{-1}Mx\|_{\infty}}{\|P^{-1}x\|_{\infty}}=\max_{y \neq 0} \frac{\|P^{-1}MPy\|_{\infty}}{\|y\|_{\infty}}=\|P^{-1}MP\|_{\infty}.
\]
An easy computation shows that for $(r,s,t) \in K$,
\[
P^{-1}A(r,s,t)P =
\begin{pmatrix}
r + \frac{bt}{2c} + \frac{cs}{2b} & - \frac{bt}{2c} + \frac{cs}{2b} \\
\frac{bt}{2c} - \frac{cs}{2b} & r - \frac{bt}{2c} - \frac{cs}{2b}
\end{pmatrix}.
\]
Recall that for a matrix $M$,
\[
\|M\|_{\infty}=\max_{i} \sum_j |M_{i,j}|,
\]
so that, since $r \geq 0$,
\[
\|P^{-1}A(r,s,t)P\|_{\infty}=r + \frac{bt}{2c} + \frac{cs}{2b} + \left | \frac{bt}{2c} - \frac{cs}{2b} \right | := F(r,s,t).
\]
It remains to find the maximum of $F$ on $K$. First, note that for $(r,s,t) \in K$,
\[
0 \leq F(r,s,t) \leq F(\a,s,t).
\]
Then, for every $(s,t) \in [0,\b] \times [0,\g]$, we can write $t = ps$, $p \geq 0$. If $p \leq c^2/b^2$, then $cs/(2b) \geq bt / (2c)$, so that $F(\a,s,t)=\a + cs/b$. But $s \leq b^2$, so $F(\a,s,t) \leq \a + bc = \a + \sqrt{\b\g}$. And if $p > c^2/b^2$, then $cs/(2b) \leq bt / (2c)$, so that $F(\a,s,t)=\a + bt/c$. But $t \leq c^2$, so $F(\a,s,t) \leq \a + bc = \a + \sqrt{\b\g}$. Finally, the maximum of $F$ on $K$, i.e. the maximum of $\|A(r,s,t)\|$ on $K$, is $\a + \sqrt{\b\g}$.

\item Assume now that $\b$ or $\g$ is zero, say e.g. $\g=0$. Take $\eps > 0$, and $M >0$ such that $\b/M < \eps$. Consider the norm $\|x\| = \|Px\|_{\infty}$, where $P$ is a diagonal matrix with diagonal $(1,M)$. For $(r,s,0) \in K$, we have as before
\[
\|A(r,s,0)\|=\|PA(r,s,0)P^{-1}\|_{\infty}=
\left \|
\begin{pmatrix}
r & s/M \\
0 & r 
\end{pmatrix}
\right \|_{\infty}.
\]
Since $s \leq \b$, this shows that $\|A(r,s,0)\| \leq \a + \eps$.
\end{enumerate}
\end{proof}

We will deal often with real-analytic functions in the remaining of the proofs. For the definitions and results on this topic, we refer to \cite{Krantz}. We will show the following result.

\begin{prop}
For $t \in [0,T_c)$ and $z \in [0,1]$, define $\phi_t(.,.,z) : [0,1]^2 \to \R^2$ by
\begin{equation}
\phi_t(x,y,z) = \left ( (1+t)x-t \d{g_0}{y}(x,y,z), (1+t)y-t \d{g_0}{x}(x,y,z) \right ),
\label{phit}
\end{equation}
and let $K_t(z)$ be the closed subset of $[0,1]^2$: $K_t(z) = \phi_t(.,.,z)^{-1} ([0,1]^2)$. Then
\begin{description}
\item[(i)] $\phi_t(.,.,z) : K_t(z) \to [0,1]^2$ is a homeomorphism. Denote $h_t(.,.,z)=(\htun(.,.,z),\htde(.,.,z))$ its inverse.
\item[(ii)] For $i=1,2$, $(x,y,z,t) \mapsto h^{(i)}_t(x,y,z)$ is an analytic function on $(0,1)^3 \times (0,T_c)$.
\end{description}
\label{propphit}
\end{prop}

\begin{proof}
\begin{description}
\item[(i)] Fix some $z \in [0,1]$ and some $t \in (0,T_c)$, and keep the notations of the statement. For notational simplicity, we omit the parameter $z$.
Let $0 \leq t < T_c$. We first want to show that $\phi_t : K_t \to [0,1]^2$ is one-to-one and onto. Fix $(u,v) \in [0,1]^2$ and let us check that there is a unique couple $(x,y) \in [0,1]^2$ such that $\phi_t(x,y)=(u,v)$. This requirement is equivalent to finding a unique fixed point to
\[
F_t(x,y)=\left ( \frac{u}{1+t} + \frac{t}{1+t} \d{g_0}{y}(x,y,z), \frac{v}{1+t} + \frac{t}{1+t} \d{g_0}{x}(x,y,z)\right ).
\]
Because of the remark above, $F_t$ is a mapping from $[0,1]^2$ to $[0,1]^2$. It remains to check that it is contracting. Its differential is
\begin{equation}
DF_t(x,y) =  \frac{t}{1+t}
\begin{pmatrix}
\ddd{g_0}{x}{y} & \dd{2}{g_0}{y} \\
\dd{2}{g_0}{x} & \ddd{g_0}{x}{y}
\end{pmatrix}
:= \frac{t}{1+t}
\begin{pmatrix}
\a(x,y,z) & \b(x,y,z) \\
\g(x,y,z) & \a(x,y,z)
\end{pmatrix}.
\label{dft}
\end{equation}
Let $\a=\a(1,1,1)=\langle ab ,c_0 \rangle$, $\b=\b(1,1,1)=\langle b^2 ,c_0 \rangle$ and $\g=\g(1,1,1)=\langle a^2 ,c_0 \rangle$. Since $t<T_c$, then $\dfrac{t}{1+t} (\a+\sqrt{\b\g}+ \eps) < 1$ for some small enough $\eps > 0$. Hence, by Lemma \ref{lemnorme}, there is a norm $\|.\|$ such that
\[
\max_{(x,y)\in[0,1]^2} \|DF_t(x,y)\| \leq \frac{t}{1+t} (\a + \sqrt{\b\g} + \eps) < 1,
\]
so that $F_t$ is contracting. Hence it has a unique fixed point. As a consequence, there is a unique couple $(x,y) \in [0,1]^2$ such that $\phi_t(x,y)=(u,v)$. Moreover, since $F_t$ is continuous with respect to $(u,v)$ and uniformly contracting in $(u,v)$, then the mapping $(u,v) \mapsto (x,y)$ is continuous, that is $h_t : [0,1]^2 \to K_t$ is a homeomorphism.

\item[(ii)] For $t_0 \in (0,T_c)$, $z_0 \in (0,1)$ and $(x_0,y_0) \in U_{t_0}$, the matrix $D \phi_{t_0}(x_0,y_0,z_0)$ is invertible. Then Theorem 2.5.3 in \cite{Krantz} shows that the inverse mapping of $\phi_t$ has real-analytic coeeficients, i.e. $h_t^{(i)}$ are real-analytic functions on $(0,1)^3 \times (0,T_c)$.
\end{description}
\end{proof}

\subsection{Study of the PDE}

The following (non-quasilinear) PDE is a central feature of our discussion
\begin{equation}
\d{g_t}{t}=\d{g_t}{x} \d{g_t}{y}-\frac{1}{1+t} \left ( x \d{g_t}{x} + y \d{g_t}{y} \right ).
\label{PDE}
\end{equation}

\noindent A preliminary result to the proof is the following. Its proof is exactly the same as the one of Lemma \ref{possol}.

\begin{lem}
Let $(c_t)_{t \in [0,T)}$ be a solution to the system
\begin{equation}
\dt c_t(p) = \frac{1}{2} \sum_{p' \preceq p} p'.(p \bsl p') c_t(p') c_t(p \bsl p') - \frac{a+b}{1+t} c_t(p)
\label{systbis}
\end{equation}
for $p=(a,b,m) \in S$, with nonnegative initial conditions. Then for all $t \in [0,T)$ and $p \in S$, $c_t(p) \geq 0$.
\label{possolbis}
\end{lem}

\begin{prop}
\begin{description}
\item[(i)] For every $z \in [0,1]$, the PDE \eqref{PDE} with initial conditions $g_0=g_0(.,.,z)$ has a unique regular (in the sense of Definition \ref{defreg}) solution $(t,x,y) \mapsto g_t(x,y,z)$ defined on $[0,T_c)\times (0,1)^2$.
\item[(ii)] The solution of the PDE is given by
\begin{equation}
g_t(x,y,z) = g_0(h_t(x,y,z),z)-\frac{t}{1+t} \d{g_0}{x}(h_t(x,y,z),z)\d{g_0}{y}(h_t(x,y,z),z),
\label{gt}
\end{equation}
where $h_t$ is defined in Proposition \ref{propphit}.
\item[(iii)] We have the alternative expression
\begin{equation}
g_t(x,y,z) = \frac{1}{1+t} \left ( \Httde(x,y,z)+ \Httun(0,y,z) \right ) + G_t(z)
\label{gtxyz2}
\end{equation}
in the notations of Theorem \ref{grosthm}.
\item[(iv)] For every $t \in [0,T_c)$, $g_t$ has an analytic expansion
\begin{equation}
g_t(x,y,z)=\sum c_t(a,b,m) x^a y^b z^m
\label{expanagt}
\end{equation}
for $(x,y,z) \in [0,1)^3$, where $c_t(a,b,m) \geq 0$.
\end{description}
\label{propPDE}
\end{prop}

\begin{rk}
Formula \eqref{gtxyz2} will be useful to compute explicit solutions, since with it, it is enough to have the analytic expansion of $h_t$ around 0 to obtain the one of $g_t$ (whose coefficients are precisely the solution to \eqref{syst}). Note however that $G_0$ may be tedious to compute in general, but since it is a function of $z$ only, it is relevant only when we wish to compute the concentrations of particles with no arms. Nonetheless, their concentrations can be obtained thanks to the sytem \eqref{syst}, since
\begin{equation}
\dt c_t(0,0,m) = \sum_{m=1}^{m-1} c_t(1,0,m') c_t(0,1,m-m').
\label{padbra}
\end{equation}
\end{rk}

\begin{proof}
We will prove the statement in three steps. Fist we will show that a solution has to be written as in \eqref{gt}. Next that this formula does provide a solution. Proving formula \eqref{gtxyz2} is then an easy matter. In all the proof, some $z \in [0,1]$ is fixed.

\begin{enumerate}
\item Let $U_t=\phi_t(.,.,z)^{-1}((0,1)^2)$, and consider $g_t$ a regular solution of \eqref{PDE} on $[0,T) \times (0,1)^2$. Fix $t_0 \in (0,T)$ and $(x,y) \in U_{t_0}$, and let
\[
(p_1(t),p_2(t)) := \left ( \d{g_t}{x}(\phi_t(x,y,z),z) , \d{g_t}{y}(\phi_t(x,y,z),z) \right ).
\]
It is easy to see that $U_t$ decreases with $t$, so for $t \leq t_0$, this definition makes sense and we can differentiate $p_i$. The regularity assumptions on $g_t$ are just those needed to allow the use of Schwarz's theorem, and an easy computation shows that, on $[0,t_0]$, $(p_1,p_2)$ solves a linear differential system with continuous coefficients, whose solution is given by \eqref{pi}. Hence
\begin{equation}
\d{g_t}{x}(\phi_t(x,y,z),z) = \d{g_0}{x}(x,y,z)/(1+t) \quad ; \quad \d{g_t}{y}(\phi_t(x,y,z),z) = \d{g_0}{y}(x,y,z)/(1+t)
\label{dgtdxphit}
\end{equation}
for all $(x,y) \in U_{t}$. Then, it is easy to check that for all $(x,y) \in U_t$, \eqref{ztpoint} and \eqref{zt} hold. Replacing $(x,y)$ by $h_t(x,y,z)$ (recall $h_t : (0,1)^2 \to U_t$ is the right-inverse of $\phi_t$), we finally obtain \eqref{gt}. This shows that the PDE has at most one solution.

\item The existence of a solution is now straightforward. Let $g_t$ be defined as in \eqref{gt}. Because of the regularity of $h_t$ and of $g_0$, $g$ has the required regularity properties. It then suffices to show that it is actually a solution. To this end, let us first compute
\[
(p_1(t),p_2(t)):= \left ( \d{g_t}{x} (\phi_t), \d{g_t}{y} (\phi_t) \right )
\]
for some fixed $t \in [0,T_c)$ and $(x,y) \in U_t$. By differentiating $g_t(\phi_t)$ with respect to $x$ and $y$, it is easy to see that it solves a linear system, which has, before $T_c$, a unique solution, given by equation \eqref{dgtdxphit}. To conclude, we may differentiate $g_t(\phi_t)$ in two different ways: one using \eqref{gt} and \eqref{dgtdxphit}. The other with the chain rule. Compounding by $h_t$ in the obtained equality readily shows that $g_t$ solves the PDE \eqref{PDE} for $t \in [0,T_c)$, $(x,y) \in (0,1)^2$.

\item The formula \eqref{gtxyz} is easy to obtain, by differentiating $g_0(h_t(x,y,z),z)$ with respect to $x$, $y$ and $z$, and using the fact that
\[
\d{g_0}{x}(h_t,z)= \httde \qquad ; \qquad \d{g_0}{y}(h_t,z)= \httun.
\]
in the notations of Theorem \ref{grosthm}.

\item To prove the last point, consider $t_0 \in [0,T_c)$. $\phi_t$ is well-defined and analytic (in $(t,x,y,z)$) in a neighbourhood of $(t_0,0,0,0)$, and $D \phi_{t_0}(0,0,0)$ is invertible. So, by theorem 2.5. in \cite{Krantz}, $h_t$ is analytic near $(t_0,0,0,0)$, hence so is $g_t=g_0(h_t)$. So we may write
\begin{equation}
g_t(x)= \sum_{(a,b,m) \in S} c_t(a,b,m) x^a y^b z^m
\label{DSEgt}
\end{equation}
for $(t,x,y,z)$ in a neighbourhood of $(t_0,0,0,0)$ and infinely differentable (even analytic) $c_t$. By analytic continuation, the $(c_t)$ are uniquely defined, so we can let
\[
E= \{ t \in [0,T_c), \forall p \in S \;\; c_t(p) \geq 0\}.
\]
By continuity, $E$ is a closed set containing $0$. On the other hand, \eqref{DSEgt} holds for $(t,x,y,z)$ in a neighbourhood of $(t_0,0,0,0)$, so for $t_0 \in E$, there is a $\eps > 0$ such that \eqref{DSEgt} holds for $t \in (t_0-\eps,t_0+\eps)$ and $(x,y,z) \in (-\eps,\eps)^3$. In particular, since $g_t$ solves the PDE \eqref{PDE}, it is easy to see, using a Cauchy product and identifying the coefficients, that $c_t$ solves \eqref{systbis} for $t \in (t_0-\eps,t_0+\eps)$. So by Lemma \eqref{possolbis}, $(t_0-\eps,t_0+\eps) \subset E$. So $E$ is open and $E=[0,T_c)$. Finally, recall from Proposition \ref{propphit} that $h_t$, and so $g_t$, are analytic on $[0,1)^3$. But we have just shown that $g_t$ has an analytic expansion around $0$ with nonnegative coefficients. So (see e.g. the proof of Berstein's theorem in \cite{Krantz}), this expression actually holds on $[0,1)^3$.
\end{enumerate}
\end{proof}

\subsection{Equivalence between the system and the PDE}

Smoluchowski's equation is solved thanks to the PDE \eqref{PDE}.

\begin{prop}
\begin{description}
\item[(i)] Let $(c_t)_{t \in [0,T)}$ be a solution to Smoluchowski's equation, and let $g_t(x,y,z) := \langle c_t , x^a y^b z^m \rangle$ be its generating function. Then for all $z \in [0,1]$, $(t,x,y) \mapsto g_t(x,y,z)$ is a regular solution to the PDE \eqref{PDE} on $[0,T \wedge \Ginf) \times (0,1)^2$, with initial conditions $g_0(.,.,z)$.
\item[(ii)] Conversely, let $(c_t(p))_{p \in S}$, $t \in [0,T)$ be a family of differentiable functions. Let $g_t(x,y,z)$ be its generating function and assume it is defined for $t \in [0,T)$, $(x,y) \in (0,1)^2$ and $z \in [0,1]$. Assume that for every $z \in [0,1]$, $g_t(.,.,z)$ is a regular solution to the PDE \eqref{PDE} with initial conditions $g_0(.,.,z)$. Then
\begin{itemize}
\item For all $p \in S$ and $t \in [0,T)$, $c_t(p) \geq 0$,
\item $(c_t)$ is a solution to Smoluchowski's equation for $t \in [T \wedge T_c)$, with initial conditions $c_0$.
\end{itemize}
\end{description}
\label{propeqsystPDE}
\end{prop}

\begin{rk}
An important feature of this result is that the PDE \eqref{PDE} and the system \eqref{syst} are equivalent only before the critical time ($T_c$ or $\Ginf$). This fact is crucial when we study the microscopic model. We indeed obtain a family of coefficients whose generating function solves the PDE (on $[0,+\infty)$), but we cannot ensure that they solve Smoluchowski's equation after $T_c$ (actually, we believe that they do not).
\end{rk}

\begin{proof}[Proof of Proposition \ref{propeqsystPDE}]
\begin{description}
\item[(i)] First note that $g$ is regular according to Remark \ref{rkgt}. If one takes $f(a,b,m) = x^a y^b z^m$ in \eqref{eqcroch}, for some fixed $(x,y,z) \in (0,1)^2 \times [0,1]$, then one gets
\[
\dt g_t(x,y,z) = \d{g_t}{x} \d{g_t}{y} - A_t y \d{g_t}{y} - B_t x \d{g_t}{x}.
\]
Recall from Lemma \ref{atbt} that when $t < \Ginf$, $A_t=B_t=1/(1+t)$. Replacing in the equation above shows that $g_t$ solves \eqref{PDE} for $(x,y) \in (0,1)^2$ and $0 \leq t < T \wedge \Ginf$.
\item[(ii)] As in the fourth part of the proof of Proposition \ref{propPDE}, we see that the $(c_t(p))$ solve \eqref{systbis}, and hence that they are nonnegative. By uniqueness of a solution to the PDE \eqref{PDE}, for $t \in [0,T \wedge T_c)$, these coefficients are those obtained in \eqref{expanagt}. Now, let $t < T_c \wedge T$, $U_t=\phi_t^{-1}(.,.,1)((0,1)^2)$, $K_t=K_t(1)=\phi_t^{-1}(.,.,1)([0,1]^2)$, and recall from \eqref{dgtdxphit} that since $g$ is a regular solution to \eqref{PDE}, then for all $(x,y) \in U_t$
\[
\d{g_t}{x} (\phi_t(x,y,1),1)=\d{g_0}{x}(x,y,1) \frac{1}{1+t},
\]
what we can write
\begin{equation}
\sum_{(a,b,m) \in S} a c_t(a,b,m) \phi_t^{(1)}(x,y,1)^a \phi_t^{(2)}(x,y,1)^b = \sum_{(a,b,m) \in S} a c_0(a,b,m) x^a y^b \times \frac{1}{1+t}.
\label{dgtdx11}
\end{equation}
Note now that since $t < T_c$, then $\phi_t(.,.,1) : K_t \to [0,1]^2$ is a homeomorphism, so $\overline{U_t}=K_t$. Since $(1,1) \in K_t$, we can pass to the limit in the equality above when $(x,y) \to (1,1)$. Using monotone convergence and the continuity of $\phi_t$, we obtain
\[
\sum_{(a,b,m) \in S} a c_t(a,b,m) = \sum_{(a,b,m) \in S} a c_0(a,b,m) \times \frac{1}{1+t} = \frac{1}{1+t}.
\]
The same reasoning shows that $\langle b,c_t \rangle = 1/(1+t)$ for $t<T_c$. Hence, we may re-write \eqref{systbis} before $T_c$ by substituting
\[
\frac{a}{1+t} = a \sum_{(a',b',m') \in S} b' c_t(a',b',m') \qquad ; \qquad \frac{b}{1+t} = b \sum_{(a',b',m') \in S} a' c_t(a',b',m'),
\]
which shows that $(c_t)$ solves Smoluchowski's equation \eqref{syst} before $T_c$.
\end{description}
\end{proof}

\subsection{Existence and uniqueness of a solution}

With these results, proving Theorem \ref{grosthm} is now an easy matter.

\begin{proof}[Proof of theorem \ref{grosthm}]
\begin{enumerate}
\item Let us first prove that $\langle a^2 + b^2, c_t \rangle$ is finite before $T_c$ and tends to $+ \infty$ when $t \to T_c$. So take $(c_t(a,b,m))_{t \in [0,T)}$ a solution to the system \eqref{syst}, and $g_t$ its generating function. Since $\langle a^2 + b^2 ,c_t \rangle < + \infty$ in a neighbourhood of 0, then we have
\[
\dd{2}{g_t}{x}(1,1,1)=\langle c_t,a^2-a \rangle,
\]
as long as $\langle a^2 + b^2,c_t \rangle < + \infty$. Note that by Lemma \ref{atbt} $\langle c_t,a \rangle$ is bounded by 1, so $\langle c_t,a^2 \rangle$ explodes if and only if $\dd{2}{g_t}{x}(1,1,1)$ explodes. Let us compute the latter. Differentiating \eqref{dgtdxphit} with respect to $x$ and $y$ and having $(x,y)$ tend to $(1,1)$, we obtain that
\[
\begin{pmatrix}
1+t-t\a & -t \g \\
-t \b & 1+t-t \a
\end{pmatrix}
\begin{pmatrix}
a \\
c
\end{pmatrix}
= \frac{1}{1+t}
\begin{pmatrix}
\g \\
\a
\end{pmatrix},
\]
where
\[
\a = \ddd{g_0}{x}{y}=\langle ab ,c_0 \rangle \quad ; \quad \b = \dd{2}{g_0}{y}=\langle b^2 - b ,c_0 \rangle \quad ; \quad \g = \dd{2}{g_0}{x}=\langle a^2 - a ,c_0 \rangle.
\]
and
\[
a = \dd{2}{g_t}{x}(1,1,1) \qquad \mathrm{and} \qquad b=\ddd{g_t}{x}{y}(1,1,1).
\]
Hence
\[
a = \dd{2}{g_t}{x}(1,1,1) = \sum_{(a,b,m) \in S} a(a-1) c_t(a,b,m) = \frac{\g}{(1+t-t \a)^2-t^2 \g \b}.
\]
This expression is valid as long as $t < T_c$, since the determinant of the matrix is then nonzero. In the same way, we also have
\[
\dd{2}{g_t}{y}(1,1,1) = \sum_{(a,b,m) \in S} b(b-1) c_t(a,b,m) = \frac{\b}{(1+t-t \a)^2-t^2 \g \b}.
\]
If $\g$ or $\b$ is nonzero, then $\langle c_t,a^2 + b^2 \rangle \to + \infty$ when $t \to T_c$. If $\g=\b=0$, then $\langle a^2 + b^2, c_t \rangle$ remains finite, but this condition also imposes that $M=1$, and so $T_c=+\infty$.

\item Uniqueness is now easy to obtain: assume $(c_t^{(1)})$ and $(c_t^{(2)})$ solve the system \ref{syst} on $[0,T)$, $T \leq T_c$, with initial conditions $c_0$. Let $g_t^{(1)}$ and $g_t^{(2)}$ be their generating functions. Since $\Ginf=T_c$ and $T \leq T_c$, then by Proposition \ref{propeqsystPDE}, for every $z \in [0,1]$, they are regular solutions to the PDE \eqref{PDE} on $[0,T) \times (0,1)^2$, with initial conditions $g_0(.,.,z)$. But by Proposition \ref{propPDE} there is a unique regular solution to the PDE on $[0,T_c)$, so $g_t^{(1)}=g_t^{(2)}$ on $[0,T)$, so that $(c_t^{(1)})=(c_t^{(2)})$.

\item The existence is given by Proposition \ref{propPDE}, (iv), and Proposition \ref{propeqsystPDE}, (ii).

\item Let us finally prove that the total mass is conserved. Consider $\psi_t(x,y,z)=(\phi_t(x,y,z),z)$, $U_t'=\psi_t^{-1}((0,1)^3)$ and $K_t'=\psi_t^{-1}([0,1]^3)$. For $(x,y,z) \in U'_t$, we can differentiate $g_t(\psi_t(x,y,z))$ with respect to $z$, and using \eqref{dgtdxphit}, we obtain
\[
\d{g_t}{z} (\psi_t(x,y,z))=\d{g_0}{z} (\psi_t(x,y,z)).
\]
Now, $\psi_t$ is a homeomorphism from $K_t'$ to $[0,1]^3$, so $\overline{U_t'}=K_t$. But $(1,1,1) \in K'_t$, so we may pass to the limit when $(x,y,z) \to (1,1,1)$ in the equality above, to obtain
\[
\d{g_t}{z}(1,1,1)=\d{g_0}{z}(1,1,1),
\]
what precisely means $\langle m,c_t \rangle=\langle m,c_0 \rangle$.
\end{enumerate}
\end{proof}

\section{Explicit formulae}

We give in this short section some explicit solutions, without giving the full details of the computations. We will always assume that at time 0, there are only particles of size 1 in the medium. So given a (finite) measure $\mu$ on $\N \times \N$, we assume
\[
c_0(a,b,m) = \mu(a,b) 1_{\{m=1\}}
\]
and as usual $\langle a,c_0 \rangle = \langle b,c_0 \rangle = 1$. To obtain the solutions, we need to invert $\phi_t$, what can be done using the (two-variable) Lagrange inversion formula (a statement is given by Good \cite{Good}). But it is much more involved than the one-dimensional formula, and the expressions it would provide can hardly be called explicit. Let us however study three easy cases. Only the last one requires the two-variable formula.

\subsection{Particles with one female arm}

The first case is when each particle has exactly one female arm, and a number of male arms distributed according to a measure $\mu_1$. So
\[
\mu(a,b) =
\left \{
\begin{array}{ll}
0  & \qquad \mathrm{if} \qquad b \neq 1 \\
\mu_1(a) & \qquad \mathrm{if} \qquad b = 1
\end{array}
\right.
\]
and we will assume that $A_0=B_0=1$, i.e. $\mu_1$ is a probability measure with unit mean. In this case, we obtain, for every $a,b \geq 0$ and $m \geq 1$
\[
c_t(a,b,m) =
\left \{
\begin{array}{lll}
0 & \qquad \mathrm{if} \qquad b \neq 1 \\
\dfrac{t^{m-1}}{(1+t)^{m+a}} \dfrac{1}{m} \dbinom{m+a-1}{a} \mu_1^{*m}(m+a-1) & \qquad \mathrm{if} \qquad b = 1 
\end{array}.
\right.
\]
In particular there exists only particles with one female arm, what is physically obvious. Moreover, the concentration $c_t(a,1,m)$ is exactly the concentration of particles with $a$ arms and mass $m$ obtained in the ``oriented model'' of \cite{B}, with initial distribution $\mu_1$. This is also natural, since in this case, $(a,1,m)$- and $(a',1,m')$-particles indeed coagulate with rate $a+a'$, which is the rate of the oriented model. Note also that $T_c = + \infty$, like in the oriented model.

\subsection{Arms with uniform random genders}

In this model, the total number of arms of a particle is chosen according to a measure $\mu_1$, then each arm is given a gender independently, with probability $1/2$. That is, we let
\[
\mu(a,b)=\mu_1(a+b) \binom{a+b}{b} \frac{1}{2^{a+b}}.
\]
We will assume that $\mu_1$ has mean 2, so that $A_0=B_0=1$. Let $\nu(j)=(j+1)\mu(j+1)$. Then we obtain, for $(a,b) \neq (0,0)$
\[
c_t(a,b,m) = \frac{1}{2} \frac{t^{m-1}}{(1+t)^{a+b+m-1}}\frac{(m+a+b-2)!}{m!a!b!} \left ( \frac{\nu_1}{2} \right )^{*m}(m+a+b-2)
\]
and
\[
c_t(0,0,m)=\frac{1}{m(m-1)} \frac{1}{2} \frac{1}{(1+1/t)^{m-1}}\left ( \frac{\nu_1}{2} \right )^{*m}(m-2)
\]
provided $\nu_1(0) > 0$. If $T_c=+ \infty$, this condition means that $\nu_1 \neq \delta_1$. In particular, one easily checks that
\[
\sum_{a+b=k} c_t(a,b,m) = 2 c_t^{sym}(k,m)
\]
where $c_t^{sym}(k,m)$ is the concentration of particles with $k$ arms and mass $m$ in the symmetric model of \cite{B}, with initial arm distribution $\mu_1/2$. The factor $2$ comes from the normalisation: in our model, the total concentration of arms in the medium is 2, when it is 1 in the symmetric model. It is also worth stressing the stronger fact that for $a,b \geq 0$, we have
\[
c_t(a,b,m)=\frac{1}{2} \binom{a+b}{b} \frac{1}{2^{a+b}} c_t(a+b,m).
\]
Hence, at any given time, the distribution of the number of male (or female) arms is still binomial. So, if at some time we chose to reassign to each arm a gender uniformly and independently, and let the system evolve on from this state, no difference would be observed. Or we could watch a system evolve like the symmetric model starting from an arm distribution $\mu_1/2$, and then at some time give the arms a gender uniformly at random and independently. The evolution afterwards will be the evolution of the sexed model with initial arm distribution $\mu$. Note as before that the critical time is the same than in the symmetric model with initial distribution $\mu_1/2$.

More generally, consider initial concentrations such that for all $a,b \geq 0$ $\mu(a,b)=\mu(b,a)$, and the solution $(c_t)_{t \in [0,T_c)}$ to Smoluchowski's equation \eqref{syst}. Then it is easy to check (by uniqueness) that for all $t \in [0,T_c)$ $c_t(a,b,m)=c_t(b,a,m)$, and that, if we denote
\[
k_t(l,m):= \sum_{a+b=l} c_t(a,b,m),
\]
then $(k_t)$ is governed (up to a factor $1/2$) by the symmetric Smoluchowski equation of \cite{B}. Hence, in this case too, $k_t(l,m)=2c_t^{sym}(l,m)$.

\subsection{Particles with one gender}
\label{onegender}

Let us finally consider the more intricate case where at time 0, the arms of each particle have all the same gender. This is motivated by the idea of ionic bonds: a particle with only male (resp. female) arms can be considered as a cation (resp. an anion), and cations can only bond with anions. Hence, consider, for $i=1,2$, $\mu_i$ two measures with mean 1 such that $\mu_1(0)=\mu_2(0)$, and take
\[
\mu(a,b)=
\left \{
\begin{array}{lll}
\mu_1(a) & \qquad \mathrm{if} \qquad & b=0 \\
\mu_2(b) & \qquad \mathrm{if} \qquad & a=0 \\
0 & \qquad \mathrm{else} &  \\
\end{array}
\right.
\]
and $\nu_i(j)=(j+1)\mu_i(j+1)$. The two-variable Lagrange inversion formula gives, for $(a,b) \neq (0,0)$
\[
c_t(a,b,m) = \dfrac{t^{m-1}}{(1+t)^{m+a+b-1}} \sum_{k=0}^m \dfrac{(m-k+b-1)!(k+a-1)!}{(m-k)!k!a!b!} \nu_1^{*(m-k)}(k+a-1) \nu_2^{*k}(m-k-1+b).
\]
If we also let
\[
\nu_1 \diamond \nu_2 (m) = (m-1) \sum_{k=1}^{m-1} \frac{1}{k} \nu_1^{*(m-k)}(k-1) \frac{1}{m-k} \nu_2^{*k} (m-k-1)
\]
then, for $m \geq 2$,
\[
c_t(0,0,m)=\frac{1}{m-1} \frac{1}{(1+1/t)^{m-1}} \nu_1 \diamond \nu_2(m),
\]
provided $\nu_1(0)\nu_2(0) > 0$ (which means, if $T_c = + \infty$, that $\nu_1$ and $\nu_2$ are not $\delta_1$). In particular, we see that if $T_c=+ \infty$ and $\nu_1,\nu_2 \neq \delta_1$, then for $m \geq 2$,
\[
c_t(a,b,m) \to
\left \{
\begin{array}{ll}
0 & \mathrm{if} \; (a,b) \neq (0,0) \\
\dfrac{1}{m-1} \nu_1 \diamond \nu_2(m) & \mathrm{if} \; a=b=0.
\end{array}
\right.
\]
Hence all the arms are used to coagulate. Chemically, this means that there are no more ions in the medium. The limiting distribution of the sizes is given by $(m-1)^{-1} \nu_1 \diamond \nu_2 (m)$. We will generalize this fact in the following section, and give a probabilistic interpretation of the measure $\nu_1 \diamond \nu_2$. Also, if $M_i$ is the mean of $\nu_i$, then $M = \sqrt{M_1 M_2}$ and $T=1/(M-1)$, or $+ \infty$ if $M \leq 1$. If $\mu_1=\mu_2=\mu$, then the critical time is the same as in the symmetric model with initial distribution $\mu$.

\section{Limiting concentrations and Galton-Watson processes}

\subsection{Convergence of the concentrations}

In this section, we will study the limiting concentrations. Similarly to what happens in the oriented and symmetric model of \cite{B}, we expect the concentrations to converge when the time tends to $+ \infty$, whenever gelation does not occur. Physically, this would mean that the system converges to a terminal state where all arms have been used (otherwise, further coagulations ``should'' occur). This is actually true, and this is an easy consequence of the preceeding results.
\begin{cor}
Assume $T_c = + \infty$, and let $(c_t)_{t \geq 0}$ be the solution to Smoluchowski's equation \eqref{syst}.
\begin{description}
\item[(i)] When $t \to + \infty$, there exists limiting concentrations $\cinf(m)$ such that
\[
c_t(a,b,m) \to \cinf(m) 1_{\{a=b=0\}}
\]
in $\ell^1(\N \times \N \times \N^*)$.
\item[(ii)] For $z \in [0,1)$, the generating function $\ginf(z)$ of $(\cinf(m))_{m \geq 1}$ is the antiderivative vanishing at $0$ of
\[
\d{g_0}{z} \left ( \hinfun(z), \hinfde(z),z \right ),
\]
where $(\hinfun,\hinfde)$ is characterised by
\begin{equation}
\left \{
\begin{array}{lll}
\hinfun(z) & = & \d{g_0}{y} \left ( \hinfun(z),\hinfde(z),z) \right ) \\
\hinfde(z) & = & \d{g_0}{x} \left ( \hinfun(z),\hinfde(z),z) \right ).
\end{array}
\right.
\label{eqhinf}
\end{equation}
\end{description}
\label{cinf}
\end{cor}

\begin{proof}
\begin{description}
\item[(i)] Since $\Ginf=T_c=+\infty$, then \eqref{eqatbt} holds for all $t \geq 0$, so
\[
\sum_{(a,b) \neq (0,0)} c_t(a,b,m) \underset{t \to +\infty}{\longrightarrow} 0.
\]
Then, using \eqref{padbra}, we get for all $t \geq 0$
\[
c_t(0,0,m)=c_0(0,0,m) + \sum_{m'=1}^{m-1} \int_0^t c_s(1,0,m')c_s(0,1,m-m') \ds.
\]
But the integrand is bounded by $A_s B_s = 1/(1+s)^2$. Hence the integral has a finite limit when $t \to + \infty$, and so does $c_t(0,0,m)$. Finally, $\langle m,c_t \rangle$ is bounded by Theorem \ref{grosthm}, and $\langle m, \cinf \rangle < + \infty$ by Fatou's lemma, so the Cauchy-Schwarz inequality shows that
\[
\sum_{m \geq 1} |c_t(0,0,m)-\cinf(m)| \underset{t \to +\infty}{\longrightarrow} 0
\]
and the result follows.
\item[(ii)] By $\ell^1$-convergence, we have
\[
\ginf(z) = \lim_{t \to +\infty} g_t(0,0,z),
\]
so, using \eqref{gtxyz} and the fact that $\Httun$ and $\Httde$ are bounded by 1,
\[
\ginf(z) = \lim_{t \to + \infty} G_t(z).
\]
It just remains to check that $\htun(0,0,z)$ and $\htde(0,0,z)$ do have a limit when $t \to + \infty$. From their definition \eqref{defhtt}, they have the same limit (if any) as $\ktun(z) := \httun(0,0,z)$ and $\ktde(z) := \httde(0,0,z)$. Now fix $Z \in [0,1)$, and consider $\ktun$ and $\ktde$ as (continuous) maps on $[0,Z]$. But $h_t$ is the right-inverse of $\phi_t$, so
\begin{equation}
\d{g_0}{x}(\ktun(z),\ktde(z),z)=\ktde(z) \qquad ; \qquad \d{g_0}{y}(\ktun(z),\ktde(z),z)=\ktun(z)
\label{eqkt}
\end{equation}
and, since $\langle a^2 + b^2 , c_0 \rangle < + \infty$, $\d{g_0}{x}(x,y,z)$ has a bounded differential on $[0,1]^2 \times [0,Z]$. Hence $\ktde$, and for the same reason $\ktun$, are Lipschitz-continuous on $[0,Z]$, with a constant independent of $t$. Ascoli's theorem thus shows that the families $(\ktun)$ and $(\ktde)$, $t \geq 0$ lie in a compact set (for the uniform topology on $[0,Z]$). So the family $(\ktun,\ktde)$ lies in a compact set, and passing to the limit in \eqref{eqkt} shows that any of its limit points solves \eqref{eqhinf}. But since $T_c=+ \infty$, the application
\[
(x,y) \mapsto \left ( \d{g_0}{y}(x,y,z) , \d{g_0}{x}(x,y,z) \right )
\]
is contracting for every $z \in [0,Z]$. So there is a unique solution to \eqref{eqhinf}, and $(\ktun,\ktde)$ converges to this solution.
\end{description}
\end{proof}

\subsection{Connection with two-type Galton-Watson processes}

In \cite{B}, Bertoin shows that for monodisperse initial conditions (i.e. $c_0(a,m)=\mu(a)1_{ \{m=0 \}}$ for some measure $\mu$) and when gelation does not occur, the limiting concentrations can be described in terms of Galton-Watson processes. The same kind of analogy is observed in our case. Precisely, consider a Galton-Watson tree with two genders, constructed as follows. We start from a male or a female ancestor. It gives birth to a number $a$ of male children, and a number $b$ of female children, where $(a,b)$ is distributed according to a law $\mu_m(a,b)$ if the ancestor is a male, $\mu_f(a,b)$ if it is a female. Then each child gives birth to a certain number of children, distributed according to $\mu_m$ or $\mu_f$, depending on his gender, and so on. Consider $T_f(\mu_m,\mu_f)$ (resp. $T_m(\mu_m,\mu_f)$) the total population of such a Galton-Watson process, starting from a female (resp male) ancestor. Let for $r \in [0,1)$
\[
g_f(r)=\E(r^{T_f}) \qquad ; \qquad g_m(r)=\E(r^{T_m})
\]
their generating functions. It is an easy exercise to check that they solve the following system, where $\phi_f$ (resp. $\phi_m$) is the generating function of $\mu_f$ (resp. $\mu_m$)
\begin{equation}
\left \{
\begin{array}{lll}
g_m(r) & = & r \phi_m(g_m(r),g_f(r)) \\
g_f(r) & = & r \phi_f(g_m(r),g_f(r)).
\end{array}
\right.
\label{eqgmgf}
\end{equation}
Besides, if $T(\mu_m,\mu_f)$ is the size of a Galton-Watson tree started from a male and a female ancestor (each tree growing independently), then
\begin{equation}
\P(T(\mu_m,\mu_f)=m)=[z^m]g_m(r)g_f(r),
\label{pgmgf}
\end{equation}
where $[z^m]h(z)$ is the coefficient of $z^m$ in the expansion around 0 of an analytic function $h$.

Now, let us go back to our study. Assume monodisperse initial conditions, i.e. there is a finite measure $\mu$ on $\N \times \N$ such that $c_0(a,b,m)=1_{\{m=0\}} \mu(a,b)$, and assume $T_c=+ \infty$. We will use the same notations as in the previous section. Using \eqref{eqhinf}, we obtain
\[
\ginf'(z)=\frac{1}{z} \hinfun(z) \hinfde(z),
\]
so that for $n \geq 2$
\begin{equation}
[z^n]g_{\infty}(z)=\frac{1}{n-1}[z^n]\hinfun(z)\hinfde(z).
\label{ginfhinf}
\end{equation}
Let
\begin{itemize}
\item $\nu_m(a,b)=(b+1)\mu(a,b+1)$ the probability measure with generating fuction $\phi_m:=\d{g_0}{y}$,
\item $\nu_f(a,b)=(a+1)\mu(a+1,b)$ the probability measure with generating fuction $\phi_f:=\d{g_0}{x}$
\end{itemize}
(they are probability measures because of the assumption $\langle a, c_0 \rangle = \langle b, c_0 \rangle = 1$), and consider the two-type Galton-Watson process as above with these reproduction laws. Because of \eqref{eqhinf} and \eqref{eqgmgf}, $(g_m,g_f)$ and $(\hinfun,\hinfde)$ solve the same equation, which has a unique solution by Corollary \ref{cinf}, so $\hinfun=g_m$, $\hinfde=g_f$. Hence by \eqref{pgmgf} and \eqref{ginfhinf}
\[
\P(T(\nu_m,\nu_f) = m) = [z^m]\hinfun(z) \hinfde(z) = (m-1) [z^m] g_{\infty}(z) = (m-1) \cinf(m).
\]
Finally, let us call a measure $\mu$ degenerate if $\mu = \delta_{(1,1)}$ or $\mu = \dfrac{1}{2} (\delta_{(2,0)} + \delta_{(0,2)})$, or $\mu(a,b)=0$ for $a \neq 1$, or $\mu(a,b)=0$ for $b \neq 1$. We let the reader check (using e.g. Theorem 10.1 in \cite{Har}) that under the assumptions $T_c = + \infty$, and ruling out the degenerate cases, $T_m(\nu_m,\nu_f)$ and $T_f(\mu_m,\mu_f)$ are finite a.s., and that the process is supercritical if $T_c < + \infty$.
\begin{cor}
The limiting concentrations verify for $m \geq 2$
\[
\P(T(\nu_m,\nu_f) = m) = (m-1) \cinf(m).
\]
Moreover, if $\mu$ is not degenerate, then $T(\nu_m,\nu_f) < + \infty$ a.s.
\label{corGW}
\end{cor}
Hence, the law $\nu_1 \diamond \nu_2$ defined in Section \ref{onegender} is the law of the total population of a two-type Galton Watson process started from one male and one female ancestors, where the males give birth to females according to the law $\nu_2$, and the females give birth to males according to the law $\nu_1$. In particular, if $\nu_1=\nu_2=\nu$, then $\nu \diamond \nu$ is the distribution of the size of a Galton-Watson tree with reproduction law $\nu$ and starting from two ancestors. So we get (what is not obvious from the formula for $\diamond$), that for $m \geq 2$
\[
\nu \diamond \nu (m) = \frac{2}{m} \nu^{*m}(m-2).
\]
This corollary answers another question about gelation. By Theorem \ref{grosthm}, the total mass $\langle m , c_t \rangle$ is conserved as time passes, so gelation does not occur before $T_c$. But if $T_c = + \infty$, it may occur at infinity: some mass may be lost then. For monodisperse initial conditions, Corollary \ref{corGW} proves that this cannot happen, except in the degenerate cases. Denote indeed $C_t = \langle c_t , 1 \rangle$. Because of the $\ell^1$-convergence in Corollary \ref{cinf}, $C_t \to \sum \cinf(m)$, and Equation \eqref{eqcroch} yields $C_t=C_0-t/(1+t)$. Hence
\[
C_0 - 1= \sum_{m \geq 1} \cinf(m).
\]
So, whenever $\mu$ is not degenerated, Corollary \ref{corGW} gives
\begin{align*}
\sum_{(a,b,m) \in S} m \cinf(m) = & \sum_{(a,b,m) \in S} \left ( (m-1) \cinf(m) + \cinf(m) \right )=\P(T(\mu_m,\mu_f) < + \infty)+ \sum \cinf(m)\\
= & 1 + \sum \cinf(m) = 1+ C_0 - 1 = C_0,
\end{align*}
which is precisely the total mass at time 0. For the degenerate cases, we can get explicit expressions for the concentrations, and these show that the mass at infinity is 0.

\section{Microscopic model}

\subsection{Notations and preliminary results}

The goal of this section is to construct a sequence of random processes modeling the coagulation of particles with male and female arms. We will start with $n$ particles (and then let $n \to + \infty$). Let us first set some notations.
\begin{itemize}
\item Recall $S=\N \times \N \times \N^*$.
\item $\lb 0,n \rb=\{0,\dots,n\}$.
\item $M>0$ is a fixed real number. The number of arms and the total mass are assumed to grow at most like $Mn$ (see the definition of $E_n$).
\item For $p=(a,b,m) \in S$ and $p'=(a',b',m') \in S$, we will denote $p.p'=a'b+ab'$ the \textit{rate} of coagulation, and $p \circ p'=(a+a'-1,b+b'-1,m+m')$ the type of the particle resulting from such a coagulation.
\item The sequence of the number of $p$-particles is an element of
\[
E_n= \left \{ N \in \lb 0, n \rb^S, \sum_{(a,b,m) \in S} (a+b+m) N(a,b,m) \leq M n \right \}
\]
which is a finite set.
\item $\dfrac{1}{n}E_n$ is a subset of
\[
E= \left \{ C \in [0,1]^S, \sum_{(a,b,m) \in S} (a+b+m) C(a,b,m) \leq M \right \}.
\]
An element of $E$ represents the sequence of concentrations of $p$-particles. $E$ is a metric space endowed with the distance
\[
d(C^{(1)},C^{(2)})=\sum_{p \in S} \left | C^{(1)}(p) - C^{(2)}(p) \right |.
\]
\item We will call C-convergence the compact convergence (i.e. uniform convergence on every compact set) for functions from $\R^+$ to $E$.
\item $\DH$ is the space of c\`adl\`ag functions from $[0,+ \infty)$ to a metric space $(H,d)$, endowed with the Skorokhod distance. We will call S-convergence the convergence for Skorokhod's distance. For the basic facts about Skorokhod distance for functions with value in a (complete separable) metric space, see \cite{EK}.
\end{itemize}
It is easy to check the following result.
\begin{lem}
$(E,d)$ is a compact metric space. In particular, it is a Polish space.
\label{Ecomp}
\end{lem}

\subsection{Model}

Let us now introduce the model. Informally, we consider a finite number $n$ of particles with integer mass, and assume that at time 0, the total mass of the system plus the total number of arms is less than $Mn$. Then, each pair formed of a $p$-particle and of a $p'$-particle may coagulate with rate $\frac{1}{2} p.p'$, independently of the other pairs, to form a $p \circ p'$-particle, that is, the time one has to wait to see them coagulate is exponential with parameter $\frac{1}{2} p.p'$. In other words, assume the system in in the state $\eta$ at a given time, that is $\eta \in E_n$ and $\eta(p)$ is the number of $p$-particles. There are $\eta(p) \eta(p')$ (or $\eta(p) (\eta(p) -1)$ if $p=p'$) pairs formed of a $p$-particle and of a $p'$-particle. Let
\[
\le(p,p') =
\left \{
\begin{array}{lll}
\frac{1}{2} p.p'\eta(p)\eta(p') & \quad \mathrm{if} \quad p \neq p' \\
\frac{1}{2} p.p\eta(p) (\eta(p) -1) & \quad \mathrm{if} \quad p=p'.
\end{array}
\right.
\]
We set independently on each couple $(p,p')$ an exponential clock with parameter $\le(p,p')$ (an exponential random variable with parameter 0 is assumed to be a.s. infinite). There is a.s. one and only clock which rings first. If it is the clock on the couple $(p,p')$, then the system jumps to the state $\eta + \Dp$ where
\[
\begin{array}{ll}
\Dp(p) = \Dp(p') = -1 & \quad \mathrm{if} \quad p \neq p' \\
\Dp(p) = -2 & \quad \mathrm{if} \quad p=p' \\
\Dp(p \circ p') = +1.
\end{array}
\]
Then restart the construction afresh from the new state. Note that only finitely many $\eta(p)$ are nonzero, so the first jump occurs after an exponential time with parameter
\[
\le = \sum_{p,p' \in S} \le(p,p') < + \infty.
\]
We will consider the Markov chain constructed according to this rule. That is, we fix for every $n \geq 1$
\begin{itemize}
\item An element $\Xn_0$ of $E_n$, which is the initial number of particles.
\item A pure-jump Markov process $\Xn$ on $E_n$, defined on some probability space $(\Omega^n,\mathcal{A}^n,\P^n)$, starting from $\Xn_0$, and with generator
\[
Gf(\eta)=\sum_{(p,p') \in S^2} ( f(\eta + \Dp) - f(\eta) ) \le(p,p')
\]
for every bounded function $f \; : \; E_n \to \R$. The construction of such a process is obvious since $E_n$ is finite.
\item The rescaled and time-changed process
\[
\Cn_t=\frac{1}{n}\Xn_{t/n}.
\]
\end{itemize}
Note that $\Cn$ is a pure-jump Markov process on $\frac{1}{n} E_n \subset E$, starting from $\Cn_0=\Xn_0/n$, and with generator
\[
\Gn f(\eta)=\sum_{(p,p') \in S^2} \left ( f \left ( \eta + \frac{1}{n} \Dp \right ) - f(\eta) \right ) \len(p,p')
\]
where
\[
\len(p,p')=\frac{1}{n} \l_{n \eta} (p,p').
\]
The law $P_n$ of the process $\Cn$ is a probability measure on $\DE$. We will prove that the sequence $(P_n)$ is tight, and that for every limit point $P$, and almost every process $(C_t)$ with law $P$, $(C_t)$ solves some system, which is Smoluchowski's equation \eqref{syst} before the critical time. Because of the uniqueness of such a solution, this will show that $(P_n)$ itself converges to the solution of Smoluchowski's equation before the critical time. The proof of tightness is analoguous to the one in \cite{Jeon}, up to some slight modifications.

\subsection{Tightness}

\begin{lem}
The sequence $(P_n)_{n \geq 0}$ is tight.
\label{Pntendue}
\end{lem}

\begin{proof}
We will use the classical tightness criterion stated in \cite{JM}, page 34, or in \cite{EK}, Theorem 7.2. For $t \geq 0$, let $P_t^{(n)}$ be the law of $\Cn_t$, which is a probability measure on $E$. Since $E$ is compact by Lemma \ref{Ecomp}, the tightness of the sequence $(P_t^{(n)})_{n \geq 0}$ is obvious.

Now, $\Cn$ is a pure-jump procees on $\frac{1}{n}E_n \subset \E$, with generator $\Gn$. Hence, when the process is in the state $\eta$, then the time before the next jump is exponential with parameter
\[
\len := \sum_{p,p' \in S} \len(p,p') = \frac{1}{2} \left ( \sum_{p, p' \in S} n p.p' \eta(p) \eta(p') - \sum_{p \in S} p.p \eta(p) \right )
\]
and, since $\eta \in E$, $\len \leq M^2 n := cn$. Now take $N > 0$, $\b > 0$, $\eps > 0$, and let $\dl > 0$ such that $N = \dl l$ for some $l \in \N^*$, and $3 c \dl e / \b < 1$. Define now
\[
w^N(Y,\dl) := \inf_{\pi \in \Pi_{\dl}} \max_{t_i \in \pi} \sup_{t_i \leq s < t < t_{i+1}} d(Y_t,Y_s)
\]
$\Pi_{\dl}$ being the set of all subdivisions $0 = t_0 < t_1 < \dots < t_n = N$ of $[0,N]$ such that $t_{i+1} - t_i \geq \dl$ for all $i$. Consider the partition $t_0=0 < t_1 = \dl < \dots < t_l= \l = N$ of $[0,N]$. Let $Z_i := \sup \limits_{t_i \leq s < t < t_{i+1}} d(\Cn_s,\Cn_t)$ for $0 \leq i \leq l-1$. Then
\[
\P_n(w^N(\Cn, \dl) > \b) \leq \P_n \left ( \max_{0 \leq i \leq l-1} Z_i > \b \right ) \leq l \max_{0 \leq i \leq l-1} \P_n(Z_i > \b).
\]
But the size of a jump, that is $d(\Cn_{t^-},\Cn_t)$, is $3/n$. Hence, if $Z_i > \b$, then the process has jumped more than $k:=\lceil \b n / 3 \rceil := \lceil c' n \rceil$ times between $t_i$ and $t_{i+1}$ (where $\lceil x \rceil$ is the first integer strictly greater than $x$). If $S_k$ if the time of the $k$-th jump, the Markov property tell us that
\[
\P_n(Z_i > \b) \leq \P_n(S_k \leq \dl).
\]
But $S_k$ is the sum of $k$ independent exponential random variables, with parameter smaller than $cn$. So, if $S'_k$ is the sum of $k$ independent exponential random variables with parameter $cn$ (on a probability space $(\Omega,\mathcal{A},\P)$), then $S_k$ is stochastically dominated by $S'_k$, that is
\[
\P_n(S_k \leq \dl) \leq \P(S'_k \leq \dl).
\]
To conclude, note that the last term is the probability that a Poisson process with parameter $cn$ jumps more than $k$ times on $[0, \dl]$, and Stirling's formula shows that this tends to zero for $(c \dl e / c') < 1$.
\end{proof}

\subsection{Convergence}

In this section, we prove the convergence of $\Cn$ to a process solving the system \eqref{systbis}, and deduce that it solves Smoluchowski's equation \eqref{syst} before $T_c$.

\begin{prop}
Assume that the following convergences in distribution hold
\begin{itemize}
\item For every $p \in S$, $\Cn_0(p) \to c_0(p)$ for some non random $c_0(p) \geq 0$,
\item $\sum \limits_{(a,b,m) \in S} a \Cn_0(a,b,m) \to 1$,
\item $\sum \limits_{(a,b,m) \in S} b \Cn_0(a,b,m) \to 1$.
\end{itemize}
Let $P$ be a limit point of $(P_n)$, and let $(c_t)$ be a process with law $P$. Then a.s. $(c_t)$ solves the system \eqref{systbis}, with initial conditions $(c_0)$.
\label{propconv}
\end{prop}

In the following proofs, we take a subsequence of $(\Cn)$ which converges in law to some possibly random $c \in \DE$. For notational simplicity, we will assume that $(\Cn)$ itself tends to $c$. Since $E$ is compact, it is separable, and hence, so is $\DE$. Skorokhod's representation theorem (cf e.g. \cite{EK}) now allows us to assume that the $\Cn$ are defined on the same probability space $(\Om,\mathcal{F},\P)$, that $\Cn \to c$ a.s. (that is, for almost every $\om \in \Om$, the function $\Cn(\om)$ tends to $c(\om)$ for Skorokhod's distance), and that in the statement, there is a.s. convergence. We will also constantly use the fact that for every bounded Borel function $f : E \to \R$, the processes
\begin{equation}
\Cn_t-\Cn_0-\int_0^t \Gn f(\Cn_s) \; \ds := \Mn_t
\label{mt}
\end{equation}
and
\begin{equation}
\left( \Mn_t \right )^2 - \int_0^t \left ( \Gn(f^2)(\Cn_s)-2f(\Cn_s) \Gn f(\Cn_s) \right ) \ds
\label{mt2}
\end{equation}
are martingales. Note also that if $f : E \to \R$ is ``linear'', then for all $\eta \in E$
\begin{equation}
\Gn f (\eta) = \frac{1}{2} \sum_{p,p' \in S} f \left ( \Dp \right ) p.p' \eta(p) \eta(p') - \frac{1}{2n} \sum_{p \in S} f \left ( \Dl_{p,p} \right ) \eta(p)
\label{Gnf}
\end{equation}
and that
\begin{equation}
\left ( \Gn (f^2) - 2 f \Gn f \right ) (\eta) = \frac{1}{2n} \left ( \sum_{p,p' \in S} f \left ( \Dp \right )^2 \eta(p) \eta(p') p.p' - \frac{1}{n} \sum_{p \in S} f \left ( \Dl_{p,p} \right )^2 \eta(p) p.p \right )
\label{Gnf2}
\end{equation}
We will also need the following convergence result.

\begin{lem}
Let
\[
\An_s = \sum_{(a,b,m) \in S} a \Cn_s(a,b,m) \qquad \mathrm{and} \qquad \Bn_s = \sum_{(a,b,m) \in S} b \Cn_s(a,b,m).
\]
Then $(\An)$ and $(\Bn)$ C-converge a.s. to $t \mapsto 1/(1+t)$.
\label{convAnBn}
\end{lem}

\begin{proof}
Obviously, we cannot pass to the limit immediately in these expressions. So consider the maps from $E$ to $\R$: $C \mapsto \langle a, C \rangle$ and $C \mapsto \langle b, C \rangle$, which are measurable and bounded (by $M$). By (\ref{mt}), there are martingales $\MAn$ and $\MBn$ such that
\begin{equation}
\An_t=\An_0 - \int_0^t \An_s \Bn_s \; \ds + \MAn_t \quad ; \quad \Bn_t=\Bn_0 - \int_0^t \An_s \Bn_s \; \ds + \MBn_t.
\label{AkBk}
\end{equation}
Now, \eqref{mt2} and \eqref{Gnf2} show that the quadratic variation of $\MAn$ verifies
\[
\left \langle \MAn \right \rangle_t \leq \frac{1}{n} \int_0^t A_s B_s \; \ds \leq \frac{M^2t}{n}.
\]
By Doob's inequality,
\[
\E \left ( \left ( \sup \limits_{0 \leq t \leq T} \MAn_t \right )^2 \right ) \to 0
\]
for all $T > 0$. Hence there is a subsequence of $(\MAn)$ which C-converges a.s. to 0. In particular, it S-converges. For notational simplicity, we will assume that $(\MAn)$ itself converges. For the same reason, we may assume that $(\MBn)$ also S-converges a.s. to 0. 

Now, note that the proof of Lemma \ref{Pntendue} still works for $(\An)$, since the size of its jumps is bounded by $1/n$. So $(\An)$ is tight, and by Prokhorov's theorem, this means that for almost every $\om \in \Omega$, $(\An(\om))_{n \geq 0}$ lies in a compact of $\DE$ (actually, this is a consequence of the proof of Lemma 9, and we do not need this implication of Prokhorov's theorem). The same works for $(\Bn)$, so we can find $\Omega' \subset \Omega$ with $\P(\Omega')=1$, such that for every $\om \in \Omega'$, $C_n(\om) \to c(\om)$ (for Skorokhod's distance), $\MAn(\om) \to 0$ and $\MBn(\om) \to 0$ (compactly), and $(\An(\om))_{n\geq 0}$ and $(\Bn(\om))_{n\geq 0}$ lie in a compact.

Next, fix $\om \in \Omega'$, and let us find the limit of $(\An(\om),\Bn(\om))$ (for the product topology --- which is not Skorokhod's topology on $\mathbb{D}([0,+ \infty),E^2)$). Since it lies in a compact set, it is enough to show that it has only one limit point. So assume $(\An(\om),\Bn(\om))$ converges to some $(A,B)$. Then $(\An_t(\om))$ converges to $A_t$ for every $t \in K$, the set of continuity points of $A$. But $A$ is c\`adl\`ag, so it has only countably many points of discontinuity, and hence $K^c$ has Lebesgue-measure 0. Hence $(\An(\om))$ converges to $A$ Lebesgue-a.s., and ditto for $(\Bn(\om))$. Also, $(\An(\om))$ and $(\Bn(\om))$ are bounded by $M$, so using dominated convergence in (\ref{AkBk}) and recalling that $\An_0(\om)$ and $\Bn_0(\om) \to 1$ by assumption, we obtain
\[
A_t= 1 -  \int_0^t A_s B_s \ds \qquad ; \qquad B_t = 1 - \int_0^t A_s B_s \ds.
\]
Hence
\[
A_t=B_t=\frac{1}{1+t}.
\]
Finally there is only one limit point to $(\An(\om),\Bn(\om))$. So $(\An(\om))$ and $(\Bn(\om))$ both S-converge to $t \mapsto 1/(1+t)$, and, since this function is continuous, they C-converge.
\end{proof}

\begin{rk}
As pointed out in the proof, the convergence of $\An$ and $\Bn$ to the actual number of arms
\[
A_t=\sum_{(a,b,m) \in S} a c_t(a,b,m) \qquad \mathrm{and} \qquad B_t=\sum_{(a,b,m) \in S} b c_t(a,b,m)
\]
is not obvious. There is no such problem for a strictly sublinear coagulation rate (as in Jeon's proof \cite{Jeon}). In our (linear) case, we prove below that this convergence holds before the critical time (we also refer to Norris \cite{Norris} for general sublinear rates in a model with no arms). In fact, if there is a solution $(c_t)$ to \eqref{syst} defined after $T_c$, we believe that $\An_t$ and $\Bn_t$ do not converge to $A_t$ and $B_t$ after $T_c$ (and that this number of arms is then stricly lesser than $1/(1+t)$). This would suggest that $T_c$ is actually a gelation time: some of the arms are lost in a ``gel'' (a particle with an infinite mass and infinitely many arms).
\end{rk}

\begin{proof}[Proof of Proposition \ref{propconv}]
\begin{enumerate}
\item Take some $p_0=(a_0,b_0,m_0) \in S$, and let for $C \in \DE$, $f(C)=C(p_0)$. According to (\ref{mt}),
\begin{equation}
\Cn_t(p_0)-\Cn_0(p_0)-\int_0^t \Gn f(\Cn_s) \; ds := \Mnp_t
\label{Cnt}
\end{equation}
is a martingale. Note also that for $p,p' \in S$, $f(\Dl_{p,p'})$ is 0, except if $p$ or $p'$ or $p \circ p'$ is $p_0$. Hence, it is easy to check using (\ref{mt}) that
\begin{equation}
\begin{split}
\Gn f(\Cn_s)= & - \sum_{p \in S} \Cn_s(p) \Cn_s(p_0) p_0.p + \frac{1}{2} \sum_{p \preceq p_0}  p.(p_0 \bsl p) \Cn_s(p)\Cn_s(p_0 \bsl p) \\
 & - \frac{1}{n} \sum_{p \in S} f \left ( \Dl_{p,p} \right ) \Cn_s(p).
\end{split}
\label{Gn}
\end{equation}
The last term is due to the difference between $\le(p,p')$ when $p \neq p'$ and when $p=p'$. In any case, it tends to 0 uniformly on $\R^+$ and uniformly in $p_0$.

\item Let us now study the martingale term. By Doob's inequality, we have for every $T > 0$
\[
\E \left ( \left ( \sup_{0 \leq t \leq T} \Mnp_t \right )^2 \right ) \leq 4 \E \left ( \left ( \Mnp_T \right ) ^2 \right )
\]
and by (\ref{mt2}), this last term is
\[
\E \left ( \int_0^T (\Gn f^2 - 2 f \Gn f) (\Cn_s) \ds \right ).
\]
But by (\ref{Gnf2}), and since $f \left ( \Dp \right ) \leq M$ for all $p,p' \in S$, then $(\Gn f^2 - 2 f \Gn f) (\Cn_s) \leq M^4/n$, so that
\[
\E \left ( \left ( \sup_{0 \leq t \leq T} \Mnp_t \right ) ^2 \right ) \to 0.
\]
Hence, there is a subsequence of $(\Mnp)$ which a.s. converges to 0 uniformly on $\R^+$. For notational simplicity, we will now assume that $(\Mnp)$ itself C-converges to 0. Using the diagonal method, we may as well assume that $(\Mnp)$ C-converges to 0 for every $p_0 \in S$.

\item We have already seen in the proof of Lemma \ref{Pntendue} that $d(\Cn_t,\Cn_{t^-}) \leq 3 / n$ a.s. By continuity of $X \mapsto \sup_{s\in [0,t]} d(X_{s^-},X_s)$ (cf \cite{EK}), this ensures that $c$ is almost surely continuous , so $\Cn$ actually C-converges to $c$. From the definition of $d$, it is also obvious that $\Cn(p)$ C-converges to $c(p)$ for every $p \in S$.

\item With these results, we may now pass to the limit in \eqref{Cnt} and \eqref{Gn}. Write \eqref{Gn} in the form
\[
\Gn f (\Cn_s)=-\a_n(s)+\frac{1}{2} \b_n(s) + \eps_n(s).
\]
Equation \eqref{Cnt} shows that
\begin{equation}
\Cn_t(p_0) = \Cn_0(p_0) - \int_0^t \a_n(s) \ds + \frac{1}{2} \int_0^t \b_n(s) \ds + \int_0^t \eps_n(s) \ds + \Mn_t.
\label{systn}
\end{equation}
By Point 3, $\Cn_t(p_0)$ C-converges a.s. to $c(p_0)$. $\Cn_0(p_0)$ tends to $c_0(p_0)$ by assumption. $\b_n(t)$ is a finite sum, so
\[
\b_n(t) \to \sum_{p \leq p_0} p.(p_0 \bsl p) \Cn_s(p) \Cn_s(p_0 \bsl p)
\]
compactly. Finally, note that
\begin{align*}
\a_n(s) = & \Cn_s(p_0) \left ( a_0 \sum_{(a,b,m) \in S} b \Cn_s(a,b,m) + b_0 \sum_{(a,b,m) \in S} a \Cn_s(a,b,m) \right )\\
 = & \Cn_s(p_0) (a_0 \Bn_s + b_0 \An_s).
\end{align*}
By Lemma \ref{convAnBn}, $\An_t$ and $\Bn_t$ converge compactly to $t \mapsto 1/(1+t)$, so
\[
\lim_{n \to + \infty} \a_n(s) = \frac{a+b}{1+t} c_s(p_0) \quad \mathrm{a.s.}
\]
compactly. Since these are all compact convergences, we can pass to the limit in (\ref{systn}), for all $p_0 \in S$. This readily shows that $(c(p))$ solves \eqref{systbis}, with initial conditions $(c_0)$.
\end{enumerate}
\end{proof}

\begin{thm}
Assume the same hypothesis as in Proposition \ref{propconv}, and assume as well 
\[
\sum \limits_{(a,b,m) \in S} (a^2+b^2) c_0(a,b,m) < + \infty.
\]
Let $T_c$ be defined as in Definition \ref{Tc}. Then $(\Cn_t)_{t \in [0,T_c)}$ converges (in distribution) to the unique solution of Smoluchowski's equation \eqref{syst}.
\end{thm}

\begin{rk}
Obviously, convergence has to be understood with respect to Skorokhod's topology on $[0,T)$ (which is the trace topology of Skorokhod's topology on $[0,+ \infty)$). In particular the sequence of the laws of $(\Cn_t)_{t \in [0,T)}$ is tight.
\end{rk}

\begin{proof}
Let $Q_n$ be the law of $(\Cn_t)_{t \in [0,T_c)}$. The sequence $(Q_n)$ is tight. Let $Q$ one of its limit points, and let $c$ a process with law $Q$. By Proposition \ref{propconv} above, $c$ solves a.s. the system \eqref{systbis}, with initial conditions $(c_0)$. Now, let $g_t(x,y,z)$ the (a priori random) generating function of $c$. It is easy to see that $g$ is well defined for $(t,x,y,z) \in [0,T_c) \times (0,1)^2 \times [0,1]$, and that $g_t(.,.,z)$ is regular for every $z \in [0,1]$. Moreover, we see as in the proof of Proposition \ref{propeqsystPDE} that for every $z \in [0,1]$, $g_t(.,.,z)$ solves the PDE \eqref{PDE} with initial conditions $(x,y) \mapsto g_0(x,y,z)=\sum c_0(a,b,m) x^a y^b z^m$. Hence by Proposition \ref{propeqsystPDE}, $(c_t)$ solves Smoluchowski's equation \ref{syst} until $T_c$. But by Theorem \ref{grosthm}, there is a unique solution to this equation on $[0,T_c)$. Hence there is a unique limit point to $(Q_n)$, so the sequence itself converges to the solution of Smoluchowski's equation on $[0,T_c)$.
\end{proof}

\noindent \textbf{Acknowledgments} This is a part of the author's PhD thesis. I would like to thank my advisors Jean Bertoin and Lorenzo Zambotti for introducing this subject, for their useful advice, and for their encouragement. My thanks also to the referees for their careful check and advice.

\end{document}